\documentclass[12pt]{elsarticle}
\usepackage[margin=2.5cm]{geometry}
\usepackage{lipsum}

\makeatletter
\def\ps@pprintTitle{%
   \let\@oddhead\@empty
   \let\@evenhead\@empty
   \def\@oddfoot{\reset@font\hfil\thepage\hfil}
   \let\@evenfoot\@oddfoot
}
\makeatother

\usepackage{amsfonts,color,morefloats,pslatex}
\usepackage{amssymb,amsthm, amsmath,latexsym}

\newtheorem{theorem}{Theorem}
\newtheorem{lemma}[theorem]{Lemma}
\newtheorem{proposition}[theorem]{Proposition}

\newtheorem{definition}[theorem]{Definition}
\newtheorem{example}[theorem]{Example}

\newtheorem{remark}[theorem]{Remark}

\newcommand{\rank}{{\mathrm{rank}}}

\newcommand{\tr}{{\mathrm{Tr}}}

\newcommand{\gf}{{\mathbb{F}}}

\newcommand{\wt}{{\mathtt{wt}}}

\newcommand{\cC}{{\mathcal{C}}}
\newcommand{\cD}{{\mathcal{D}}}

\newcommand{\ba}{{\mathbf{a}}}

\newcommand{\bc}{{\mathbf{c}}}
\newcommand{\bg}{{\mathbf{g}}}

\newcommand{\bzero}{{\mathbf{0}}}

\usepackage{enumerate}
\usepackage{booktabs}
\usepackage{arydshln}
\allowdisplaybreaks[4]
\begin{document}

\begin{frontmatter}

\title{Self-orthogonal codes from plateaued functions}

\tnotetext[fn1]{
This research was supported in part by the National Natural Science Foundation of China under Grant 12271059, in part by the open research fund of National Mobile Communications Research Laboratory of Southeast University under Grant 2024D10 and in part by the Shaanxi Fundamental Science Research Project for Mathematics and Physics (Grant No. 23JSZ008).
}

\author[author1,author2]{Peng Wang}
\ead{wp20201115@163.com}
\author[author1,author2]{Ziling Heng\corref{cor}}
\ead{zilingheng@chd.edu.cn}

\cortext[cor]{Corresponding author}
\address[author1]{School of Science, Chang'an University, Xi'an 710064, China}
\address[author2]{National Mobile Communications Research Laboratory, Southeast University, Nanjing 211111, China}

\begin{abstract}
Self-orthogonal codes are of interest as they have important applications in quantum codes, lattices and many areas.
In this paper, based on the weakly regular plateaued functions or plateaued Boolean functions, we construct a family of linear codes with four nonzero weights.
This family of linear codes is proved to be not only self-orthogonal but also optimally or almost optimally extendable.
Besides, we derive binary and ternary linearly complementary dual codes (LCD codes for short) with new parameters from this family of codes.
Some families of self-dual codes are also obtained as byproducts.
\end{abstract}

\begin{keyword}
Weakly regular plateaued function, self-orthogonal code, LCD code, optimally extendable code

\MSC  94B05 \sep 94A05

\end{keyword}

\end{frontmatter}
\section{Introduction}\label{sec1}
Let $\gf_q$ be the finite field with $q$ elements, where $q$ is a  power of a prime $p$. Let $\gf_q^*:=\gf_q\backslash \{0\}$.
\subsection{Linear codes and self-orthogonal codes}
Let $n, k, d$ be three positive integers. For a nonempty set $\cC\subseteq \gf_q^n$,  if $\cC$ is a $k$-dimensional linear subspace of $\gf_q^n$, then it is called an $[n,k,d]$ linear code over $\mathbb{F}_q$, where $d$ denotes its minimum distance. Let $A_{i}$  be the number of codewords with weight $i$ in a linear code $\cC$ of length $n$, where $0 \leq i \leq n$. The sequence $(1,A_{1},A_{2}, \cdots ,A_{n})$ is called the weight distribution of $\cC$ and the polynomial $A(z)=1+A_{1}z+A_{2}z^2+ \cdots +A_{n}{z^n}$ is referred to as the weight enumerator of $\cC$.
Weight distribution is a significant research subject as it not only describes the error detection and error correction abilities of the code, but also can be used to calculate the error probabilities of error detection and correction of the code. The weight distributions of some linear codes were studied in the literature \cite{D4, HZ3, SJ, WY3}.

Optimal linear codes play an important role in both theory and practice. An $[n, k, d]$ linear code $\cC$ over $\gf_q$ is said to be optimal if there does not exist $[n, k, d + 1]$ linear code over $\gf_q$  and almost optimal if there exists an $[n, k, d + 1]$ optimal code over $\gf_q$. The following is the well-known sphere-packing bound on the parameters of codes.
\begin{lemma}\cite[Sphere-packing bound]{H}\label{sphere}
  Let $M$ be the maximum number of codewords in a code over $\gf_q$ of length $n$ and minimum distance $d$. Then
  \begin{eqnarray*}
    M \leq \frac{q^m}{\sum\limits_{i=0}^{t}\tbinom{m}{i}(q-1)^i},
  \end{eqnarray*}
  where $t=\lfloor(d-1)/2\rfloor$ and $\lfloor \cdot \rfloor$ denotes the floor function.
\end{lemma}

Define the dual code of an $[n,k]$ linear code $\cC$ by
$$
\mathcal{C}^{\perp}=\left\{ \mathbf{u} \in \mathbb{F}_{q}^{n}: \langle  \mathbf{u}, \mathbf{v} \rangle=\mathbf{0} \mbox{ for all }\mathbf{v} \in \mathcal{C} \right\},
$$
where $\langle \cdot \rangle$ denotes the standard inner product. Then $\cC^{\perp}$ is a linear code with length $n$ and dimension $n-k$.
If $\cC \subseteq \cC^{\perp}$, then $\cC$ is called a self-orthogonal code.
In particular, if $\cC=\cC^{\perp}$, then $\cC$ is said to be self-dual.
Self-orthogonal codes have  nice applications in lattices \cite{W}, linear complementary dual codes (LCD codes for short) \cite{MA2}, quantum codes \cite{CSS} and so on.
For a given $q$-ary linear code $\cC$, it is natural to study whether it is self-orthogonal or not.
If $q=2,3$, the following provides simple conditions for a linear code to be self-orthogonal by the divisibility of its weights.
\begin{lemma}\cite[Theorems 1.4.8 and 1.4.10 ]{H}\label{lem-23}
Let $\cC$ be a linear code over $\gf_p$. For $p=2$, if every codeword of $\cC$ has weight divisible by four, then $\cC$ is self-orthogonal.~For $p=3$, $\cC$ is self-orthogonal if and only if every codeword of $\cC$ has weight divisible by three.
\end{lemma}

Recently, Li and Heng established a sufficient condition for a $q$-ary linear code containing the all-1 vector to be self-orthogonal if $q$ is a power of an odd prime \cite{HZ5}.
\begin{lemma}\cite{HZ5}\label{lem-self}
Let $q$ be a power of $p$. where $p$ is an odd prime. Let $\mathcal{C}$ be an $[n,k,d]$ linear code over $\gf_q$ with  $\mathbf{1} \in \cC,$ where $\mathbf{1}$ is all-1 vector of length $n$. If $\cC$ is $p$-divisible, then $\cC$ is self-orthogonal.
\end{lemma}

The following lemmas show that we can derive self-dual codes from self-orthogonal codes. These self-dual codes are the subcodes of the duals of the self-orthogonal codes.
\begin{lemma}\cite[Proposition 4.5]{Be}\label{lem-dual1}
Let $q$ be a power of $2$. A linear code $\cC \subseteq \mathbb{F}_{q}^{n}$  of  dimension $k\geq\frac{n}{2}$ contains a self-dual subcode if and only if $n$ is even and $\cC^{\perp}$ is self-orthogonal, that is, $\cC^{\perp} \subseteq \cC.$
\end{lemma}

\begin{lemma}\cite[Proposition 4.6]{Be}\label{lem-dual}
Let $\cC \subseteq \mathbb{F}_{q}^{n}$ be a linear code of  dimension $k\geq\frac{n}{2}$.
\begin{enumerate}[(1)]
\item  If $q\equiv 1 \pmod {4}$, then $\cC$ contains a self-dual subcode if and only if $n$ is even and $\cC^{\perp}$ is self-orthogonal.
 \item If $q\equiv 3 \pmod {4}$, then $\cC$ contains a self-dual subcode if and only if $n \equiv 0 \pmod {4}$ and $\cC^{\perp}$ is self-orthogonal.
 \end{enumerate}
\end{lemma}
\subsection{LCD codes}
 An $[n, k, d]$ linear code $\cC$ is called an Euclidean LCD code if $\cC \cap \cC^{\perp} = \{\mathbf{0}\}$, where $\mathbf{0}$ denotes the zero vector of length $n$.
 In \cite{M}, Massey introduced the definition and gave the algebraic characterization of LCD codes. He also proved the existence of
asymptotically good LCD codes and provided an optimum linear coding solution for the two-user binary adder channel. In \cite{SN}, Sendrier showed that LCD codes satisfy the asymptotic Gilbert-Varshamov bound. LCD codes have nice applications in lattices, network coding, and multisecret-sharing schemes \cite{AA, MB, HO}.
Besides, LCD codes also have application in counter-measures to passive and active side channel analyses on embedded cryptosystems.
 Many families of LCD codes have been constructed in \cite{WY2, LKZ}. The existence of  $q$-ary $(q >3)$ Euclidean LCD codes has been studied in \cite{CMT}. However, the existence of binary and ternary Euclidean LCD codes has not been totally investigated.

Self-orthogonal codes can be used to construct LCD codes.
It is known that a generator matrix $G$ of an $[n,k,d]$ self-orthogonal code satisfies $GG^{T}=\mathbf{0}_{k,k}$, where $\mathbf{0}_{k,k}$ denotes the zero matrix of size $k \times k$. In other words, $G$ is row-self-orthogonal.
The following lemma shows how to construct leading-systematic codes from a row-self-orthogonal matrix, where
a linear code is said to be leading-systematic if and only if it has a generator matrix of the form $G=[I: P]$ for an identity matrix $I$.
\begin{lemma}\cite{MA2}\label{lem-LCD}
A leading-systematic linear code $\cC$ with systematic generator matrix $G := [I : P]$ is an LCD code if (but not only if) the matrix $P$ is row-self-orthogonal, where $I$ denotes an identity matrix.
\end{lemma}

Note that a self-orthogonal code has many different generator matrices. If $P_1$ and $P_2$ are two different generator matrices of a self-orthogonal code $\cC$, then $G_1 := [I : P_1]$ and $G_2 := [I : P_2]$ may generate two inequivalent LCD codes. Finding a suitable generator matrix $P$ such that $G = [I : P]$ generate an LCD code with large minimum distance, whose dual also have large minimum distance, is important.

\subsection{Optimally and almost optimally extendable codes}
 In \cite{AB}, Bringer et al. introduced a direct sum masking (DSM) countermeasure to oppose side channel attack (SCA) and fault injection attack (FIA)
  which have brought great threats to the implementation of block cipher. It needs two linear codes $\cC$ and $\cD$ satisfying $\cC \bigoplus\cD =\gf_q^n$, where $\cC$ and $\cD$ are used to encode the sensitive data and encode random data, respectively. In particular, if $\cD=\cC^\perp$, then $\cC$ is an LCD code.
  For a general pair $(\cC,\cD)$ satisfying $\cC \bigoplus\cD =\gf_q^n$, it is desirable to change $\cC$ and $\cD$ into $\cC'$ and $\cD'$ to  protect the sensitive data stored in the register from SCA and FIA, where $\cC'$ is obtained by adding $k$ bits of zeros to the end of each codeword in $\cC$
  and $\cD'$ is derived by appending the identity matrix at the end of the generator matrix of $\cD'$. If $d({\cD'}^{\perp})=d(\cD^{\perp})$, then $\cD$ is said to be \emph{optimally extendable}.  $\cD$ is said to be \emph{almost optimally extendable} if $d(\cD'^{\perp})=d(\cD^{\perp})-1$.

In the literature, there are only a few known constructions of optimally or almost optimally extendable codes. Carlet, Li and Mesnager gave two families of optimally extendable linear codes and two families of almost optimally extendable linear codes from the first-order Reed-Muller codes and irreducible cyclic codes of dimension two \cite{Carlet3}.
Quan, Yue and Hu studied (almost) optimally extendable linear codes from irreducible cyclic codes, MDS codes and NMDS codes \cite{YueQ}.
In \cite{WXY1}, Wang, Heng, Li and Yue constructed some families of almost optimally extendable codes from self-orthogonal codes.
In \cite{HL}, Heng, Li, Wu and Wang constructed two families of  optimally or almost optimally extendable codes from some functions over finite fields.
\subsection{The purpose of this paper}
The purpose of this paper is to construct linear codes from weakly regular plateaued functions over finite fields. Let $p$ be an odd prime and $f(x)$ be a function from $\gf_q$ to $\gf_p$ with $f(0)=0$. We define a $p$-ary linear code as
\begin{eqnarray}\label{eq-Cfbar}
{\overline{\cC_f}}=\left\{\bc_{{(a,b,c)}}=(af(x)+\tr_{q/p}(bx)+c)_{x \in \gf_q}: a \in \gf_p, b \in \gf_q, c \in \gf_p\right\}.
\end{eqnarray}
We remark that $\overline{\cC_f}$ is the augmented code of the extended code of
\begin{eqnarray}\label{eq-Cf*}
\cC_f^*=\left\{\bc_{(a,b)}=(af(x)+\tr_{q/p}(bx))_{x \in \gf_q^*}: a \in \gf_p, b \in \gf_q\right\}.
\end{eqnarray}
 When $f$ is a weakly regular bent function, the code $\overline{\cC_f}$ was studied in \cite{HL}.

In this paper, we mainly study the parameters and weight distribution of $\overline{\cC_f}$ and  prove the following  if $f$ is a weakly regular plateaued function:
\begin{enumerate}
\item[\small$\bullet$] $\overline{\cC_f}$ is self-orthogonal;
\item[\small$\bullet$] $\overline{\cC_f}$ is optimally or almost optimally extendable if we  select a suitable generator matrix of $\overline{\cC_f}$.
\end{enumerate}
Besides, we also derive ternary LCD codes from $\overline{\cC_f}$.
Moreover,  if $f$ is a plateaued Boolean function, we also prove that $\overline{\cC_f}$ is self-orthogonal and determine their parameters and weight distribution.
The binary code $\overline{\cC_f}$ is proved to be almost optimally extendable.
We also use the binary code $\overline{\cC_f}$ to  construct a family of binary LCD codes.
Finally, we obtain some families of self-dual codes from the duals of the self-orthogonal codes in this paper.

\section{Preliminaries}\label{sec2}
In this section, we present some mathematical foundations which will be used to give the main results in this paper.

\subsection{Characters and Gaussian sums over finite fields}
Let $q=p^m$ with $p$ a prime. Let $\zeta_p$ denote the primitive $p$-th root of complex unity.
 Define an additive character $\chi$ of $\gf_q$ by the homomorphism from the additive group $\gf_q$ to the complex multiplicative group $\mathbb{C}^*$ such that
$$
\chi(x+y)=\chi(x)\chi(y), ~ x, y \in \gf_q.$$
For any $a \in \gf_q$,
an additive character of $\gf_q$ is defined by the function $\chi_a(x)=\zeta_p^{\tr_{q/p}(ax)},\ x\in \gf_q$, where $\tr_{q/p}(x)$ is the trace function from $\gf_q$ to $\gf_p$. By definition, we have $\chi_a(x)=\chi_1(ax)$. Particularly, $\chi_0$ is called the trivial additive character of $\gf_q$ and $\chi_1$ is called the canonical additive character of $\gf_q$. The following is the orthogonal relation of additive characters (see \cite{Lidl}):
\begin{eqnarray*}
\sum_{x\in \gf_q}\chi_a(x)=\begin{cases}
q    &\text{if $a=0,$ }\\
0     &\text{if $a\in\gf_q^*.$}
\end{cases}
\end{eqnarray*}

Let $\alpha$ be a primitive element of $\gf_q$. The multiplicative character $\psi$ of $\gf_q$ is defined as the homomorphism from the multiplicative group $\gf_q^*$ to the complex multiplicative group $\mathbb{C}^*$ such that
$$
\psi(xy)=\psi(x)\psi(y), ~x, y \in \gf_q^*.$$
The function
$\psi_j(\alpha^k)=\zeta_{q-1}^{jk}$, where $k=0,1,\cdots,q-2$,
is a multiplicative character for each $j=0, 1,..., q-2$. In particular, $\psi_0$ is called the trivial multiplicative character and $\eta:=\psi_{\frac{q-1}{2}}$ is referred to as the quadratic multiplicative character of $\gf_q$ if $q$ is odd. The following is the orthogonal relation of multiplicative characters (see \cite{Lidl}):
\begin{eqnarray*}
\sum_{x\in\gf_q^*}\psi_j(x)=\begin{cases}
q-1    &\text{if $j=0,$ }\\
0     &\text{if $j \neq 0$.}
\end{cases}
\end{eqnarray*}

For an additive character $\chi$ and a multiplicative character $\psi$ of $\gf_q$, the \emph{Gaussian sum} $G(\psi, \chi)$ over $\gf_q$ is defined by
$$G(\psi, \chi)=\sum_{x\in \gf_q^*}\psi(x)\chi(x).$$
Specially, $G(\eta,\chi)$ is called the \emph{quadratic} \emph{Gaussian sum} over $\gf_q$ for nontrivial $\chi$.
The explicit value of quadratic \emph{Gaussian sum} is given in the following.
\begin{lemma}\cite[Theorem 5.15]{Lidl}\label{quadGuasssum1}
Let $q=p^m$ with $p$ odd. Then
\begin{eqnarray*}
G(\eta,\chi_1)=(-1)^{m-1}(\sqrt{-1})^{(\frac{p-1}{2})^2m}\sqrt{q}.
\end{eqnarray*}
\end{lemma}

\subsection{Weakly regular plateaued functions and the cyclotomic fields}
Let $f(x)$ be a $p$-ary function from $\gf_{p^m}$ to $\gf_p$. Define the Walsh transform of $f(x)$ as follows:
\begin{eqnarray*}
\text{W}_f(\beta):=\sum_{x \in \gf_{p^m}}\zeta_p^{f(x)-\tr_{p^m/p}(\beta x)},\ \beta \in \gf_{p^m}.
\end{eqnarray*}
Then $f$ is said to be \emph{balanced} over $\gf_p$ provided that $f$ takes every value of $\gf_p$  exactly $p^{m-1}$ times. Otherwise, $f$ is said to be \emph{unbalanced}. It is known that $f$ is balanced if and only if $\text{W}_f(0)=0$. A function $f$ is referred to as a \emph{bent} function if $\mid \text{W}_f(\beta) \mid^2=p^{m}$ for any $\beta \in \gf_{p^m}$. Besides, $f$ is said to be \emph{$s$-plateaued} provided that $\mid\text{W}_f(\beta)\mid^2\in \{{0,p^{m+s}}\}$ for any $\beta \in \gf_{p^m}$, with $0 \leq s \leq m$. For an $s$-plateaued function $f$, its \emph{Walsh support} is defined by the set $S_f=\{\beta \in \gf_p^n:{\mid \text{W}_f(\beta) \mid}^2 = p^{m+s}$\}.
\begin{lemma}\cite{Mesnager1}\label{tim-plateaued}
Let $f$ be an $s$-plateaued function from $\gf_{p^m}$ to $\gf_p$. Then \ for\ $\beta \in \gf_{p^m}$, ${\mid \text{W}_f(\beta) \mid}^2$ takes the values $p^{m+s}$ and $0$ for the times $p^{m+s}$ and $p^m-p^{m-s}$, respectively.
\end{lemma}

\begin{lemma}\cite{Mesnager1}
Let $f$ be $s$-plateaued from $\gf_{p^m}$ to $\gf_p$. Define the sets $Z(\text W_f):=\{(\alpha, \beta) \in \gf_p^* \times \gf_{p^m}:\text W_f(\alpha^{-1}\beta)=0\}$ and
$S(\text W_f):=\{(\alpha,\beta) \in \gf_p^* \times \gf_{p^m}:\text W_f(\alpha^{-1}\beta)\neq 0\}$. Then, the sizes of $Z(\text W_f)$ and $S(\text W_f)$ are respectively equal to $(p-1)(p^m-p^{m-s})$ and $(p-1)p^{m-s}$.
\end{lemma}

In particular, for $p=2$, the Walsh distribution of plateaued Boolean functions is given in the following \cite{2yuan1},\cite{2yuan2}.

\begin{lemma}\label{plateaued-boolean}
Let $m+s$ be an even integer with $0\leq s \leq m$ and $f$ be an $s$-plateaued Boolean function with $f(0)=0$ ~from ~$\gf_{2^m}$ to~$\gf_2$. Then, for $\beta \in \gf_{2^m}$, the Walsh distribution of $f$ is given by
\begin{eqnarray*}\label{2-sPB}
\text{W}_f(\beta)=
\begin{cases}
2^{\frac{m+s}{2}},   &\text{$2^{m-s+1}+2^{\frac{m-s-2}{2}}$ times,}\\
0,   &\text{$2^{m}-2^{m-s}$ times,}\\
-2^{\frac{m+s}{2}},   &\text{$2^{m-s-1}-2^{\frac{m-s-2}{2}}$ times.}
\end{cases}
\end{eqnarray*}
\end{lemma}

\begin{definition}\label{de-sp}
Let $p$ be an odd prime and $f$ be a $p$-ary $s$-plateaued function from $\gf_{p^m}$ to $\gf_p$, where $s$ is an integer with $0\leq s\leq m$. Then $f$ is said to be weakly regular provided that there exists a complex number $u$ with unit magnitude (in fact, $\mid u\mid=1$ and $u$ does not depend on $\beta$) such that
\begin{eqnarray}\label{eq-Wf}
\text{W}_f(\beta) \in \left\{0,u{p}^{\frac{m+s}{2}}\zeta_p^{f^*(\beta)}\right\},
\end{eqnarray}
for all $\beta \in \gf_{p^m}$, where $f^*$ is a $p$-ary function over $\gf_{p^m}$ with $f^*(\beta)=0$ for every $\beta \in \gf_{p^m}\setminus S_f$. Otherwise, $f$ is said to be non-weakly regular.
\end{definition}
\begin{lemma}\label{lamwf}
Let $f$ be a weakly regular s-plateaued function and $p^*=(-1)^{\frac{p-1}{2}}p$. Then, for all $\beta \in S_f$, we have
\begin{eqnarray}\label{eq-Wf}
\text{W}_f(\beta)=\varepsilon \sqrt{p^*}^{m+s}\zeta_p^{f^*(\beta)},
\end{eqnarray}
where $\varepsilon=\pm1$ is the sign of $W_f$ and $f^*$ is a p-ary function over $S_f$.
\end{lemma}

Note that the weakly regular $0$-plateaued function is the weakly regular bent function.
Let $f$ be a weakly regular $s$-plateaued function satisfying two homogenous conditions: $f(0)=0$ and $f(a x)=a^hf(x)$ for all $x \in \gf_{p^m}$ and $a \in \gf_p^*$, where $h$ is an even positive integer with $\gcd(h-1,p-1)=1.$
Let \text{WRP} and \text{WRPB} denote the sets of such weakly regular plateaued unbalanced and balanced functions, respectively.

\begin{lemma}\cite{Sinak}\label{WRP0WRPB}
Let $f\in \text{WRP}$ or $f\in \text{WRPB}$. Then $f^*(0)=0$ and $f^*(a\beta)=a^lf^*(\beta)$ ~for all $a \in \gf_p^*$ and $\beta \in S_f$, where $l$ is an even positive integer with $gcd(l-1,p-1)=1$. For any $\beta \in S_f$ (respectively, $\beta \in \gf_{p^m}\backslash S_f$), we have $z\beta \in S_f$ (respectively, $z\beta \in \gf_{p^m}\backslash S_f$) for every $z \in \gf_p^*$.
\end{lemma}

To give the sufficient and necessary condition for a quadratic function to be $s$-plateaued,
we now give a brief introduction on the quadratic functions (see \cite{Helleseth}). Recall that any quadratic function from $\gf_{p^m}$ to $\gf_p$ with no linear term can be represented by
\begin{eqnarray*}\label{q-fun}
Q(x)=\sum_{i=0}^{\lceil m/2 \rceil}\tr^m({a_i}x^{p^i+1}),
\end{eqnarray*}
where $\lceil\cdot\rceil$ denotes ceiling function
 and $a_i \in \gf_{p^m}$ for $0 \leq i \leq \lceil m/2\rceil$. Then there exists an $m \times m$ symmetric matrix $A$ such that $Q(x)=x^TAx$.
Let $L$ be the corresponding linearized polynomial over $\gf_{p^m}$  defined as
\begin{eqnarray*}\label{L-fun}
L(z)=\sum_{i=0}^{\lceil m/2 \rceil}({a_i}z^{p^i}+{a_i}^{p^{m-i}}z^{p^{m-i}}).
\end{eqnarray*}
Then $\rank(A)=m-s_L$ by \cite[Proposition 2.1]{HouXD} , where $0\leq s_L\leq m$ is the dimension of the following $\gf_p$-linear subspace of $\gf_{p^m}$:
\begin{eqnarray*}\label{ker-fun}
\ker_{\gf_p}(L)=\{z\in {\gf_{p^m}}:Q(z+y)=Q(z)+Q(y),\forall y \in \gf_{p^m}\}.
\end{eqnarray*}

\begin{proposition} \cite{Mesnager2}
Any quadratic function $Q$ is $s$-plateaued if and only if $\rank(A)=m-s$, i.e. $s_L=s$.
\end{proposition}

By \cite[Proposition 1]{Helleseth} and \cite[Theorem 4.3]{Cesmelioglu},  we have the following fact.

\begin{proposition}
Any quadratic function is weakly regular plateaued. Namely, there is no quadratic non-weakly regular plateaued functions.
\end{proposition}

\subsection{Cyclotomic field}
The following lemma gives some results on the cyclotomic field $\mathbb{Q}(\zeta_p)$.
\begin{lemma}\label{lem-cyclo}\cite{Mesnager1}
Let $K=\mathbb{Q(}\zeta_p)$ denote the $p$-th cyclotomic field over the rational number field $\mathbb{Q}$, wher $p$ is an odd prime. Then the following hold.
\begin{itemize}
  \item The ring of integers in $K$ is $O_K=\mathbb{Z}(\zeta_p)$ and $\{\zeta_p^i: 1 \leq i \leq p-1\}$ is an integral basis of $O_K$, where $\zeta_p$ is the primitive $p$-th root of complex unity.
  \item The field extension $K/\mathbb{Q}$ is Galois of degree $p-1$ and Galois group $\text{Gal}(K/\mathbb{Q})=\{\sigma_a: a \in \gf_p^*\},$
  where the automorphism $\sigma_a$ of $K$ is defined by $\sigma_a(\zeta_p)=\zeta_p^a$.
  \item The field $K$ has a unique quadratic subfield $L=\mathbb{Q}(\sqrt{p^*})$. For $1 \leq a \leq p-1$, $\sigma_a(\sqrt{p^*})=\eta_0(a)\sqrt{p^*}$, where $p^*=(-1)^{\frac{p-1}{2}}p$ and $\eta_0$ is the quadratic multiplicative character of $\gf_p$. Hence, the Galois group $\text{Gal}(L/\mathbb{Q})=\{1, \sigma_\gamma\}$, where $\gamma$ is a nonsquare in $\gf_p^*$.
\end{itemize}
\end{lemma}
According to Lemma \ref{lem-cyclo}, we have
\begin{eqnarray}\label{eq-sigma}
\sigma_a(\zeta_p^b)=\zeta_p^{ab} \mbox{ and } \sigma_a(\sqrt{p^*}^m)=\eta_0^m(a)\sqrt{p^*}^m.
\end{eqnarray}
\section{Four-weight linear codes from $p$-ary weakly regular $s$-plateaued functions}\label{sec3}
In this section, if  $f(x)$ is a weakly regular $s$-plateaued function, we will prove that the code $\overline{\cC_f}$ defined in Equation (\ref{eq-Cfbar}) is a family of self-orthogonal codes and its parameters and weight distribution will be determined. Besides, we will construct a family of LCD codes from $\overline{\cC_f}$ and prove that  $\overline{\cC_f}$ is an optimally extendable code.

\subsection{The weight distribution of $\overline{\cC_f}$}
From now on, we let $\lambda_1$ denote the canonical additive character of $\gf_p$ and $\phi_1$ denote the canonical additive character of $\gf_{p^m}$.
Let $\wt(\bc)$ denote the Hamming weight of a codeword $\bc$ in a linear code.

\begin{lemma}\label{theorem}
Let $t\in \gf_p^*$, $a\in \gf_p$ and $b\in \gf_{p^m}$, where $p$ is an odd prime. Let $f(x)$ be a weakly regular $s$-plateaued function with $\varepsilon=\pm1$ the sign of its Walsh transform. Let $N_t$ denote the number of solutions in $\gf_{p^m}$ of the equation $af(x) + \tr_{p^m/p}\left(bx\right)=t$. If $m+s$ is even, then we have
\begin{eqnarray*}
N_t=
\begin{cases}
0 &\text{if $a=0,b=0$,}\\
p^{m-1}     &\text{if $a=0, b \in \gf_{p^m}^*$, or $a\in \gf_p^*$, $a^{-1}b\notin S_f$,}\\
p^{m-1}+\frac{\varepsilon(p-1)\sqrt{p^*}^{m+s}}{p}     &\text{if $a \in \gf_p^*,  ~f^*(-\frac{b}{a})=a^{-1}t$ and $a^{-1}b\in S_f$,}\\
p^{m-1}-\frac{\varepsilon \sqrt{p^*}^{m+s}}{p}     &\text{if $a \in \gf_p^*,  ~f^*(-\frac{b}{a})\neq a^{-1}t$ and $a^{-1}b\in S_f$,}
\end{cases}
\end{eqnarray*}
where $p^* = (-1)^{\frac{p-1}{2}}p$.
If $m+s$ is odd, then we have
\begin{eqnarray*}
N_t=
\begin{cases}
0 &\text{if $a=0, b=0$,}\\
p^{m-1}     &\text{if $a=0, b \in \gf_{p^m}^*$, or $a \in \gf_p^*, f^*(-\frac{b}{a})=a^{-1}t$, $a^{-1}b\in S_f$,}\\
             &\text{or $a\in \gf_p^*$, $a^{-1}b\notin S_f$,}\\
p^{m-1}+\frac{\varepsilon \sqrt{p^*}^{m+s}G(\eta_0, \lambda_1)}{p}     &\text{if $a \in \gf_p^*, f^*(-\frac{b}{a})\neq a^{-1}t$, $\eta_0(f^*(-\frac{b}{a})- a^{-1}t)=1$  }\\
             &\text{and  $a^{-1}b\in S_f$,}\\
p^{m-1}-\frac{\varepsilon \sqrt{p^*}^{m+s}G(\eta_0, \lambda_1)}{p}     &\text{if $a \in \gf_p^*, f^*(-\frac{b}{a})\neq a^{-1}t$, $\eta_0(f^*(-\frac{b}{a})- a^{-1}t)=-1$ }\\
             &\text{and $a^{-1}b\in S_f$,}
\end{cases}
\end{eqnarray*}
where  $G(\eta_0,\lambda_1) = (\sqrt{-1})^{(\frac{p-1}{2})^2}\sqrt{p}$.
\end{lemma}

\begin{proof}
Based on the orthogonal relation of additive characters, we have
\begin{eqnarray*}
N_t &=&| \{x\in\gf_{p^m}: af(x) + \tr_{p^m/p}\left(bx\right)=t\}|\\
&=&\frac{1}{p}\sum_{y \in \gf_p}\sum_{x \in \gf_{p^m}}\zeta_p^{yaf(x)+\tr_{p^m/p}(ybx)-yt}\\
&=&\frac{1}{p}\sum_{y \in \gf_p^*}\lambda_1(-yt)\sum_{x \in \gf_{p^m}}\zeta_p^{yaf(x)+\tr_{p^m/p}(ybx)}+p^{m-1}.
\end{eqnarray*}

In order to calculate the value of $N_t$, we now consider the following cases.

\noindent{Case 1:} If $a=0$, then
\begin{eqnarray*}
N_t=\frac{1}{p}\sum_{y \in \gf_p^*}\lambda_1(-yt)\sum_{x \in \gf_{p^m}}\phi_1(ybx)+p^{m-1}
=\begin{cases}
        0 & \mbox{if $b=0$}, \\
        p^{m-1} & \mbox{if $b \in \gf_{p^m}^*$}.
      \end{cases}
\end{eqnarray*}

\noindent{Case 2:} If $a\neq 0$, then
\begin{eqnarray*}
N_t&=&\frac{1}{p}\sum_{y \in \gf_p^*}\lambda_1(-yt)\sum_{x \in \gf_{p^m}}\zeta_p^{ya\left(f(x)-\tr_{p^m/p}(-\frac{b}{a}x)\right)}+p^{m-1}\\
&=&\frac{1}{p}\sum_{y \in \gf_p^*}\lambda_1(-yt)\sigma_{ya}\left(\text{W}_f\left(-\frac{b}{a}\right)\right)+p^{m-1}\\
&=&\frac{1}{p}\sum_{y \in \gf_p^*}\lambda_1(-ya^{-1}t)\sigma_{y}\left(\text{W}_f\left(-\frac{b}{a}\right)\right)+p^{m-1}.
\end{eqnarray*}
where $\sigma_y$ is the automorphism of $\mathbb{Q(}\zeta_p)$ defined by $ \sigma_y(\zeta_p)=\zeta_p^{y}$.\\
{Subcase 2.1:} If $\text{W}_f(-\frac{b}{a})=0$ (namely $a^{-1}b\notin S_f$), then
\begin{eqnarray*}
N_t=p^{m-1}.
\end{eqnarray*}
{Subcase 2.2:} If $\text{W}_f(-\frac{b}{a})\neq0$ (namely $a^{-1}b\in S_f$), then we have
\begin{eqnarray*}
\sigma_{y}\left(\text{W}_f\left(-\frac{b}{a}\right)\right)=\sigma_{y}\left(\varepsilon \sqrt{p^*}^{m+s}\zeta_p^{f^*(-\frac{b}{a})}\right)
= \varepsilon \sqrt{p^*}^{m+s}\zeta_p^{yf^*(-\frac{b}{a})}\eta_0^{m+s}(y),
\end{eqnarray*}
by Lemma \ref{lamwf} and Equation \ref{eq-sigma},  where $p^* = (-1)^{\frac{p-1}{2}}p$. Then
\begin{eqnarray*}
N_t &=&\frac{\varepsilon \sqrt{p^*}^{m+s}}{p}\sum_{y \in \gf_p^*}\lambda_1(-ya^{-1}t)\zeta_p^{yf^*\left(-\frac{b}{a}\right)}\eta_0^{m+s}(y)+p^{m-1}\\
&=& \frac{\varepsilon \sqrt{p^*}^{m+s}}{p}\sum_{y \in \gf_p^*}\lambda_1\left(y\left(f^*\left(-\frac{b}{a}\right)-a^{-1}t\right)\right)\eta_0^{m+s}(y)+p^{m-1}.
\end{eqnarray*}
{Subcase 2.2.1:} If $m+s$ is an even integer, by the orthogonality of additive characters, we have
\begin{eqnarray*}
N_t&=&\frac{\varepsilon \sqrt{p^*}^{m+s}}{p}\sum_{y \in \gf_p^*}\lambda_1\left(y\left(f^*\left(-\frac{b}{a}\right)-a^{-1}t\right)\right)+p^{m-1}\\
&=&\begin{cases}
     \frac{\varepsilon (p-1) \sqrt{p^*}^{m+s}}{p}+p^{m-1} & \mbox{if $f^*(-\frac{b}{a})=a^{-1}t$}, \\
     -\frac{\varepsilon \sqrt{p^*}^{m+s}}{p}+p^{m-1} & \mbox{if $f^*(-\frac{b}{a}) \neq a^{-1}t$}.
   \end{cases}
\end{eqnarray*}
{Subcase 2.2.2:} If $m+s$ is an odd integer, by the orthogonality of additive characters, we have
\begin{eqnarray*}
N_t&=&\frac{\varepsilon \sqrt{p^*}^{m+s}}{p}\sum_{y \in \gf_p^*}\lambda_1\left(y\left(f^*\left(-\frac{b}{a}\right)-a^{-1}t\right)\right)\eta_0(y)+p^{m-1}\\
&=&\begin{cases}
     p^{m-1} & \mbox{if $f^*(-\frac{b}{a})=a^{-1}t$}, \\
     \frac{\varepsilon \sqrt{p^*}^{m+s}}{p}G(\eta_0, \lambda_1)+p^{m-1} & \mbox{if $\eta_0(f^*(-\frac{b}{a})- a^{-1}t)=1$},\\
     -\frac{\varepsilon \sqrt{p^*}^{m+s}}{p}G(\eta_0, \lambda_1)+p^{m-1} & \mbox{if $\eta_0(f^*(-\frac{b}{a})- a^{-1}t)=-1$}.
   \end{cases}
\end{eqnarray*}
where $G(\eta_0, \lambda_1) = (\sqrt{-1})^{(\frac{p-1}{2})^2}\sqrt{p}$  by Lemma \ref{quadGuasssum1}.\\

The desired conclusions follow from summarizing the above discussions.
\end{proof}

\begin{theorem}\label{weight}
Let $q=p^m$, where $p$ is an odd prime, $s$ is an integer and $m$ is a positive integer with $m\geq 2$. Let $f(x)$ be a weakly regular $s$-plateaued function with $\varepsilon=\pm1$ the sign of its Walsh transform.~Then $\overline{\cC_f}$ defined in Equation (\ref{eq-Cfbar}) is a $[q,m+2]$ linear code over $\gf_p$  and its weight distributions are respectively listed in Table \ref{tab1} for even $m+s$ with $0\leq s\leq m$ and Table \ref{tab2}  for odd $m+s$ with $0\leq s\leq m-1$. Particularly, $\overline{\cC_f}$ is self-orthogonal if $m+s\geq3$. Besides, $\overline{\cC_{f}}^{\perp}$ has parameters $[q,q-m-2,3]$ if $m+s\geq3$ and is at least almost optimal according to the sphere-packing bound.
\end{theorem}
\begin{table}[h!]
\begin{center}
\caption{The weight distribution of $\overline{\cC_{f}}$ in Theorem \ref{weight} for even $m+s$.}\label{tab1}
\begin{tabular}{@{}ll@{}}
\toprule
Weight & Frequency  \\
\midrule
$0$ & $1$ \\
$p^m-p^{m-1}-\frac{\varepsilon (p-1)(\sqrt{p^*})^{m+s}}{p}$ & $(p-1)p^{m-s}$  \\
$p^m-p^{m-1}+\frac{\varepsilon (\sqrt{p^*})^{m+s}}{p}$ & $(p-1)^2p^{m-s}$ \\
$p^m-p^{m-1}$ & $ p(p^m-1)+p(p-1)(p^m-p^{m-s})$ \\
$p^m$ &   $p-1$\\
\bottomrule
\end{tabular}
\end{center}
\end{table}
\begin{table}[h!]
\begin{center}
\caption{The weight distribution of $\overline{\cC_{f}}$ in Theorem \ref{weight} for odd $m+s$  .}\label{tab2}
\begin{tabular}{@{}ll@{}}
\toprule
Weight & Frequency  \\
\midrule
$0$ & $1$ \\
$p^m-p^{m-1}-\frac{\varepsilon (\sqrt{p^*})^{m+s+1} }{p}$ &  $\frac{(p-1)^2p^{m-s}}{2}$ \\
$p^m-p^{m-1}+\frac{\varepsilon (\sqrt{p^*})^{m+s+1} }{p}$ &  $\frac{(p-1)^2p^{m-s}}{2}$\\
$p^m-p^{m-1}$ &  $p(p^m-1)+p(p-1)(p^m-p^{m-s})+p^{m-s}(p-1)$ \\
$p^m$ &  $p-1$ \\
\bottomrule
\end{tabular}
\end{center}
\end{table}
\begin{proof} Let $\bc_{{(a,b,c)}}=(af(x)+\tr_{q/p}(bx)+c)_{x \in \gf_q}, a \in \gf_p, b \in \gf_q, c \in \gf_p$, be a codeword in $\overline{\cC_f}$.
For $c=0$,  we have
\begin{eqnarray*}
\text{wt}(\bc_{(a,b,0)}) &=& p^m-|\{x \in \gf_q: af(x)+\tr_{q/p}(bx)=0\}| \\
   &=&  p^m-\frac{1}{p}\sum_{y \in \gf_p} \sum_{x \in \gf_q}\phi_1(ybx) \lambda_1(yaf(x))\\
   &=& p^m-p^{m-1}-\frac{1}{p}\sum_{y \in \gf_p^*}\sum_{x \in \gf_q} \lambda_1(yaf(x)+y\tr_{q/p}(bx))\\
   &=& \begin{cases}
         0, & \mbox{if $a=0, b=0$},  \\
         p^m-p^{m-1}, & \mbox{if $a=0, b \neq 0$}, \\
         p^m-p^{m-1}-\frac{1}{p}\sum_{y \in \gf_p^*}\limits\sum_{x \in \gf_q}\limits\zeta_p^{ya(f(x)-\tr_{q/p}(-\frac{b}{a}x))}, & \mbox{if $a \neq 0$}.
       \end{cases}
\end{eqnarray*}
Denote by $N_1=p^m-p^{m-1}-\frac{1}{p}\sum_{y \in \gf_p^*}\limits\sum_{x \in \gf_q}\limits\zeta_p^{ya(f(x)-\tr_{q/p}(-\frac{b}{a}x))}$. Similarly to the proof of Lemma \ref{theorem}, we have the following results.
\begin{enumerate}[(1)]
\item If $\text{W}_f(-\frac{b}{a})=0$ (i.e. $a^{-1}b\notin S_f$), then
\begin{eqnarray*}
N_1=p^m-p^{m-1}.
\end{eqnarray*}
\item If $\text{W}_f(-\frac{b}{a})\neq0$ (i.e. $a^{-1} b\in S_f$), we have
\begin{eqnarray*}
  N_1 &=& \begin{cases}
            p^m-p^{m-1}-\frac{\varepsilon(p-1)\sqrt{p^*}^{m+s}}{p} & \mbox{if $f^*(-\frac{b}{a})=0$},\\
            p^m-p^{m-1}+\frac{\varepsilon \sqrt{p^*}^{m+s}}{p} & \mbox{if $f^*(-\frac{b}{a})\neq 0$},
          \end{cases}
\end{eqnarray*}
for even $m+s$, and
\begin{eqnarray*}
  N_1 &=& \begin{cases}
            p^m-p^{m-1} & \mbox{if $f^*(-\frac{b}{a})=0$},\\
            p^m-p^{m-1}-\frac{\varepsilon \sqrt{p^*}^{m+s}G(\eta_0, \lambda)}{p} & \mbox{if $f^*(-\frac{b}{a})\neq 0$ and $\eta_0(f^*(-\frac{b}{a}))=1$},\\
            p^m-p^{m-1}+\frac{\varepsilon \sqrt{p^*}^{m+s}G(\eta_0, \lambda)}{p} & \mbox{if $f^*(-\frac{b}{a})\neq 0$ and $\eta_0(f^*(-\frac{b}{a}))=-1$},
          \end{cases}
\end{eqnarray*}
for odd $m+s$.
\end{enumerate}
Then the value of $\text{wt}(\bc_{(a,b,0)})$ directly follows.

 For $c \in \gf_p^*$, denote by $N_{(a,b,c)}=| \{x \in \gf_q:af(x)+\tr_{q/p}(bx)+c=0\} | $.
Then the Hamming weight of $\bc_{(a,b,c)}$ is $\text{wt}(\bc_{(a,b,c)})=p^m-N_{(a,b,c)}$ for $a \in \gf_p, b \in \gf_{p^m}, c \in \gf_p^*$. Besides, we can verify that $N_{(a,b,c)}$ is equal to $N_t$ in Lemma \ref{theorem} for $t=-c$. Then by Lemma \ref{theorem}, if $m+s$ is an even integer, we have
\begin{eqnarray*}
\text{wt}(\bc_{(a,b,c)})&=&p^m-N_{(a,b,c)}\\
&=&\begin{cases}
p^m &\text{if $a=0,b=0,$}\\
p^m-p^{m-1}     &\text{if $a=0, b \in \gf_{p^m}^*$, or $a \in \gf_p^* $, $a^{-1}b\notin S_f,$}\\
p^m-p^{m-1}-\frac{\varepsilon(p-1)\sqrt{p^*}^{m+s}}{p}     &\text{if $a \in \gf_p^*, f^*(-\frac{b}{a})+a^{-1}c=0$, $a^{-1}b\in S_f,$}\\
p^m-p^{m-1}+\frac{\varepsilon\sqrt{p^*}^{m+s}}{p}     &\text{if $a \in \gf_p^*, f^*(-\frac{b}{a})+ a^{-1}c \neq 0,$ $a^{-1}b\in S_f$;}
\end{cases}
\end{eqnarray*}
if $m+s$ is an odd integer, we have
\begin{eqnarray*}
\text{wt}(\bc_{(a,b,c)})&=&p^m-N_{(a,b,c)}\\
&=&\left\{\begin{array}{lll}
p^m \qquad\qquad\quad \mbox{if $a=0,b=0$,}\\
p^m-p^{m-1}    \quad \quad \mbox{if $a=0, b \in \gf_{p^m}^*$, or $a \in \gf_p^*, f^*(-\frac{b}{a})+a^{-1}c=0$, $a^{-1}b\in S_f$,} \\
\qquad\qquad\qquad \thinspace \mbox{or $a \in \gf_p^*$, $a^{-1}b\notin S_f,$}\\
p^m-p^{m-1}-\frac{\varepsilon \sqrt{p^*}^{m+s}G(\eta_0, \lambda)}{p} \\
\qquad\qquad\qquad \thinspace \mbox{if $a \in \gf_p^*, f^*(-\frac{b}{a})+ a^{-1}c \neq 0$ , $\eta_0(f^*(-\frac{b}{a})+ a^{-1}c)=1$}\\
\qquad\qquad\qquad \thinspace \mbox{and $a^{-1}b\in S_f,$}\\
p^m-p^{m-1}+\frac{\varepsilon \sqrt{p^*}^{m+s}G(\eta_0, \lambda)}{p}  \\
\qquad\qquad\qquad \thinspace \mbox{if $a \in \gf_p^*, f^*(-\frac{b}{a})+ a^{-1}c \neq 0$ , $\eta_0(f^*(-\frac{b}{a})+ a^{-1}c)=-1$}\\
\qquad\qquad\qquad \thinspace \mbox{and $a^{-1}b\in S_f$.}
\end{array} \right.
\end{eqnarray*}

Now we summarize the above discussions. If $m+s$ is an even integer, then
\begin{eqnarray*}
\text{wt}(\bc_{(a,b,c)})&=&\begin{cases}
0 &\text{if $a=0,b=0,c=0$,}\\
p^m &\text{if $a=0, b =0, c \in \gf_p^*$,}\\
p^m-p^{m-1}     &\text{if $a=0, b \in \gf_{p^m}^*, c \in \gf_p$,}\\
&\text{or $a \in \gf_p^* $, $a^{-1}b\notin S_f ,c \in \gf_p$,}\\
p^m-p^{m-1}-\frac{\varepsilon(p-1)\sqrt{p^*}^{m+s}}{p}    &\text{if $a \in \gf_p^*, f^*(-\frac{b}{a})+a^{-1}c=0$ and $a^{-1}b\in S_f$,}\\
p^m-p^{m-1}+\frac{\varepsilon \sqrt{p^*}^{m+s}}{p}     &\text{if $a \in \gf_p^*, f^*(-\frac{b}{a})+ a^{-1}c \neq 0$ and $a^{-1}b\in S_f$;}
\end{cases}
\end{eqnarray*}
if $m+s$ is an odd integer, then
\begin{eqnarray*}
\text{wt}(\bc_{(a,b,c)})&=&\begin{cases}
0 &\mbox{if $a=0,b=0,c=0$,}\\
p^m & \mbox{if $a=0, b=0, c \in \gf_p^*$,}\\
p^m-p^{m-1}  &\mbox{if $a=0, b \in \gf_{p^n}^*, c \in \gf_p$, or $a \in \gf_p^*,$ $a^{-1}b\in S_f$,}\\
&\mbox{ $f^*(-\frac{b}{a})+a^{-1}c=0$, or $a \in \gf_p^*,$ $a^{-1}b\notin S_f$, $c \in \gf_p$},\\
p^m-p^{m-1}-\frac{\varepsilon \sqrt{p^*}^{m+s}G(\eta_0, \lambda)}{p}  & \mbox{if $a \in \gf_p^*, f^*(-\frac{b}{a})+a^{-1}c \neq 0$,}\\ &\mbox{$\eta_0(f^*(-\frac{b}{a})+a^{-1}c)=1$ and $a^{-1}b\in S_f$},\\
p^m-p^{m-1}+\frac{\varepsilon \sqrt{p^*}^{m+s}G(\eta_0, \lambda)}{p} & \mbox{if $a \in \gf_p^*, f^*(-\frac{b}{a})+ a^{-1}c \neq 0$, }\\ &\mbox{$\eta_0(f^*(-\frac{b}{a})+a^{-1}c)=-1$ and $a^{-1}b\in S_f$}.
\end{cases}
\end{eqnarray*}
\indent When $m+s$ is even, we denote by $w_1=p^m$, $w_2=p^m-p^{m-1}$, $w_3=p^m-p^{m-1}-\frac{\varepsilon(p-1)\sqrt{p^*}^{m+s}}{p}$, $w_4=p^m-p^{m-1}+\frac{\varepsilon\sqrt{p^*}^{m+s}}{p}$. Next we will determine the frequency $A_{w_i}, 1 \leq i \leq 4$. It is easy to deduce that $A_{w_1}=p-1$ and $A_{w_2}=p(p^{m}-1)+p(p-1)(p^m-p^{m-s})$ by Lemma \ref{tim-plateaued}. Note that
$A_{w_3}=| \{(a,b,c) \in \gf_p^* \times \gf_{p^m} \times \gf_p: f^*\left(-\frac{b}{a}\right)+a^{-1}c=0 \ \text{and} \ a^{-1}b\in S_f\}| $. For fixed $a$ and $b$, then the equation $f^*\left(-\frac{b}{a}\right)+a^{-1}c=0$ with variable $c$ has the unique solution. Hence, $A_{w_3}=(p-1)p^{m+s}$ by Lemma \ref{tim-plateaued} and $A_{w_4}=p(p-1)p^{m+s}-A_{w_3}=(p-1)^2p^{m+s}$.

When $m+s$ is odd, we denote by $w_1=p^m$, $w_2=p^m-p^{m-1}$, $w_3=p^m-p^{m-1}-\frac{\varepsilon \sqrt{p^*}^{m+s}G(\eta_0, \lambda)}{p}$, $w_4=p^m-p^{m-1}+\frac{\varepsilon \sqrt{p^*}^{m+s}G(\eta_0, \lambda)}{p}$. Next we will determine the frequency $A_{w_i}, 1 \leq i \leq 4$. It is easy to deduce that $A_{w_1}=p-1$ and $A_{w_2}=p(p^m-1)+p(p-1)(p^m-p^{m-s})+p^{m-s}(p-1)$ by Lemma \ref{tim-plateaued}. Note that $ A_{w_3}=| \{(a,b,c) \in \gf_p^* \times \gf_{p^m} \times \gf_p: f^*\left(-\frac{b}{a}\right)+a^{-1}c\neq 0, \eta_0\left(f^*\left(-\frac{b}{a}\right)+a^{-1}c\right)=1 \ \text{and} \  a^{-1}b\in S_f\}| $.
For fixed $a \in \gf_p^*$ , $b\in \gf_{p^m}$ and $\delta \in \langle \gamma^2 \rangle$, the solution of the equation $f^*(-\frac{b}{a})+a^{-1}c=\delta$ with variable $c$ is unique. Then  $A_{w_3}=\frac{(p-1)^2p^{m+s}}{2}$ by lemma \ref{tim-plateaued}. Similarly, we deduce $A_{w_4}=\frac{(p-1)^2p^{m+s}}{2}$.\\
\indent According to the weight distribution of $\overline{\cC_f}$, we deduce that $ \overline{\cC_f} $ is $p$-divisible for $ {m+s}\geq 3 $. Then by Lemma \ref{lem-self}, $ \overline{\cC_{f}} $ is a self-orthogonal code.\\
\indent Finally, we determine the parameters of $\overline{\cC_{f}}^{\perp}$. If $m+s$ is even, by the first four Pless power moments in \cite[Page 260]{H}, we have
\begin{eqnarray*}
A_1^{\perp}=A_2^{\perp}=0,\ A_3^{\perp}=\frac{p^{m-1}(p-1)(p-2)\left(p^m-p+ \varepsilon \eta_0^{\frac{m+s}{2}}(-1)p^{\frac{m+s}{2}}(p-1)\right)}{6}.
\end{eqnarray*}
Note that $A_3^{\perp} > 0$ as $m+s \geq 4$ is an even integer. Then the parameters of $\overline{\cC_{f}}^{\perp}$ are $[q, q-m-2, 3]$.
If $m+s$ is odd, by the first four Pless power moments in \cite[Page 260]{H}, we have
\begin{eqnarray*}
A_1^{\perp}=A_2^{\perp}=0,\ A_3^{\perp}=\frac{p^{m-1}(p-1)(p-2)(p^m-p)}{6}.
\end{eqnarray*}
Note that $A_3^{\perp} > 0$ as $m+s \geq 3$ is an odd integer. Then the parameters of $\overline{\cC_{f}}^{\perp}$  are $[q, q-m-2, 3]$.
The proof is completed by the sphere-packing bound in Lemma \ref{sphere}.
\end{proof}

\subsection{The optimally extendability of $\overline{\cC_f}$}
In this subsection, we will prove that $\overline{\cC_f}$ is optimally extendable by selecting a suitable generator matrix of $\overline{\cC_f}$.
Let $\gf_{q}^*=\langle\alpha\rangle$ for $q=p^m$. Then $\{1,\alpha, \alpha^{2},\ldots, \alpha^{m-1}\}$ is a $\gf_p$-basis of $\gf_{q}$. Let $d_1,d_2,\cdots,d_{q-1}, d_q$ denote all elements of $\gf_q$, where $d_q=0$. A generator matrix $G$ of $\overline{\cC_f}$ is given by
\begin{eqnarray*}
G:=\left[
\begin{array}{cccc}
1&1&\cdots&1 \\
 f(d_1) & f(d_2) & \cdots &f(d_q)\\
 \tr_{q/p}(\alpha^{0}d_{1})& \tr_{q/p}(\alpha^{0}d_{2})& \cdots &\tr_{q/p}(\alpha^{0}d_{q}) \\
\tr_{q/p}(\alpha^{1}d_{1})& \tr_{q/p}(\alpha^{1}d_{2})& \cdots &\tr_{q/p}(\alpha^{1}d_{q}) \\
\vdots &\vdots &\ddots &\vdots \\
\tr_{q/p}(\alpha^{m-1}d_{1})& \tr_{q/p}(\alpha^{m-1}d_{2})& \cdots &\tr_{q/p}(\alpha^{m-1}d_{q})
\end{array}\right].
\end{eqnarray*}
By an elementary row transformation, we obtain another generator matrix $G_1$ of $\overline{\cC_f}$ given as
\begin{eqnarray*}\label{matrix2}
G_1:=\left[
\begin{array}{cccc}
1&1&\cdots&1 \\
 f(d_1)+1 & f(d_2)+1 & \cdots &f(d_q)+1\\
 \tr_{q/p}(\alpha^{0}d_{1})& \tr_{q/p}(\alpha^{0}d_{2})& \cdots &\tr_{q/p}(\alpha^{0}d_{q}) \\
\tr_{q/p}(\alpha^{1}d_{1})& \tr_{q/p}(\alpha^{1}d_{2})& \cdots &\tr_{q/p}(\alpha^{1}d_{q}) \\
\vdots &\vdots &\ddots &\vdots \\
\tr_{q/p}(\alpha^{m-1}d_{1})& \tr_{q/p}(\alpha^{m-1}d_{2})& \cdots &\tr_{q/p}(\alpha^{m-1}d_{q})
\end{array}\right].
\end{eqnarray*}

Let $G':=[I_{m+2,m+2}:G_1]$ generate a linear code $\overline{\cC_f}'$, where $I_{m+2,m+2}$ denotes the identity matrix of size $m+2$.

\begin{theorem} \label{theorem1}
Let $q=p^m$, where $p$ is an odd prime, $m$ is a positive integer, $0 \leq s \leq m$ is an integer such that $m+s$ is even and $m+s\geq 4$.
Let $f(x)$ be a weakly regular $s$-plateaued function with $\varepsilon$ the sign of the Walsh transform of $f(x)$
\begin{enumerate}[1)]

\item When $f(x)$ is unbalanced, then we have the following results.
 \begin{itemize}
\item[1.1)] If~ $\varepsilon (\eta_0(-1))^{\frac{m+s}{2}}=1$, then $\overline{\cC_f}'$ has parameters
  $$\left[p^m+m+2, m+2, p^m-p^{m-1}-(p-1)p^{\frac{m+s}{2}-1}+2\right].$$
\item[1.2)] If~ $\varepsilon (\eta_0(-1))^{\frac{m+s}{2}}=-1$, then $\overline{\cC_f}'$ has parameters
  $$\left[p^m+m+2, m+2, p^m-p^{m-1}-p^{\frac{m+s}{2}-1}+1\right].$$
 \end{itemize}
\item  When $f(x)$ is balanced, then we have the following results.
\begin{itemize}
\item[2.1)] If ~$\varepsilon (\eta_0(-1))^{\frac{m+s}{2}}=1$, then $\overline{\cC_f}'$ has parameters
  $$\left[p^m+m+2, m+2, \geq p^m-p^{m-1}-(p-1)p^{\frac{m+s}{2}-1}+2\right].$$
\item[2.2)] If~ $\varepsilon (\eta_0(-1))^{\frac{m+s}{2}}=-1$, then $\overline{\cC_f}'$ has parameters
  $$\left[p^m+m+2, m+2, \geq p^m-p^{m-1}-p^{\frac{m+s}{2}-1}+2\right].$$
  \end{itemize}
\end{enumerate}

\indent Moreover,  $\overline{\cC_f}'^{\perp}$ has parameters $[p^m+m+2, p^m, 3]$  and is at least almost optimal with respect to the sphere-packing bound.
\end{theorem}

\begin{proof}
 It is obvious that the length of $\overline{\cC_f}'$ is $q+m+2$ and the dimension of $\overline{\cC_f}'$ is $m+2$. It is clear that the minimum distance $d(\overline{\cC_f}')$ of $\overline{\cC_f}'$ satisfies $d(\overline{\cC_f}') \geq d(\overline{\cC_{f}})$, where $d(\overline{\cC_{f}})$ denotes  the minimum distance of $\overline{\cC_{f}}$.\\
\indent Next we prove that the minimum distance $d(\overline{\cC_f}')$ of $\overline{\cC_f}'$ satisfies that $d(\overline{\cC_f}') \geq d(\overline{\cC_{f}})+1$.
 For a codeword $\ba'$ in $\overline{\cC_f}'$, we have
\begin{eqnarray*}
\ba'&=&(c,a,a_0,a_1 \cdots, a_{m-1})G_1'\\&=&(c,a,a_0, \cdots, a_{m-1}, af(d_1)+\sum_{i=0}^{m-1}a_i\tr_{q/p}(\alpha^{i}d_{1})+c+a, \cdots, af(d_q)+\sum_{i=0}^{m-1}a_i\tr_{q/p}(\alpha^{i}d_{q})+c+a),
\end{eqnarray*}
where $c, a, a_i \in \gf_p$ for all $0 \leq i \leq m-1$. Then $\ba'=(c,a,a_0,a_1 \cdots, a_{m-1}, \ba)$, where $$\ba=(c,a,a_0,a_1 \cdots, a_{m-1})G_1=(af(x)+\tr_{p^m/p}(bx) + c+a)_{x \in \gf_q}$$ is a codeword in $\overline{\cC_f}$, where $b:=\sum_{i=0}^{m-1}a_i\alpha^i$.
Assume that $\ba'$ has the minimum nonzero weight $d'$. If $d'=d(\overline{\cC_{f}})$, then we deduce  that the first $m+2$ locations are all zero. We have $c=a=a_i=0(0\leq i \leq m-1)$ and $\ba'=\bzero$, which leads to a contradiction. Thus  $d(\overline{\cC_f}') \geq d(\overline{\cC_{f}})+1$.\\
\indent Next we consider the following cases:

{Case 1:} Let $f(x)$ be unbalanced. Then $\text{W}_f(0)=\varepsilon \sqrt{p^*}^{m+s}$ and $ 0\in S_f$.\\

{Subcase 1.1:}  Let $\varepsilon (\eta_0(-1))^{\frac{m+s}{2}}=1$.
 By Theorem \ref{weight},  we deduce that $\text{wt}(\ba)=d(\overline{\cC_f})=p^m-p^{m-1}-(p-1)p^{\frac{m+s}{2}-1}$ if and only if $a \in \gf_p^*$ , $f^*(-\frac{b}{a})+a^{-1}(c+a)=0$ and $a^{-1}b\in S_f$. If $f^*(-\frac{b}{a})+a^{-1}(c+a)=0$, we have $(b,c)\neq (0,0)$. Besides, $b=0$ if and only if $(a_0,a_1,\cdots,a_{m-1})=(0,0,\cdots,0)$. Thus $(c,a_0,a_1,\cdots,a_{m-1})\neq (0,0,0,\cdots,0)$ and
 $d(\overline{\cC_f}') \geq d(\overline{\cC_{f}})+2$. Specially, if $a \in \gf_p^*$, $b=0$ and $c=-a$, then $f^*(-\frac{b}{a})+a^{-1}(c+a)=0$ and $\text{wt}(\ba')=\text{wt}(\ba)+2=d(\overline{\cC_{f}})+2$.
Hence $d(\overline{\cC_f}')=d(\overline{\cC_{f}})+2$.\\

{Subcase 1.2:} Let $\varepsilon (\eta_0(-1))^{\frac{m+s}{2}}=-1$. By Theorem \ref{weight},  we deduce that $\text{wt}(\ba)=d(\overline{\cC_f})=p^m-p^{m-1}-p^{\frac{m+s}{2}-1}$ if and only if $a \in \gf_p^*$ , $f^*(-\frac{b}{a})+a^{-1}(c+a)\neq0$ and $a^{-1}b\in S_f$.
If $a\neq 0$, $b=0$ and $c=0$, then $f^*(-\frac{b}{a})+a^{-1}(c+a)=1$ and $\text{wt}(\ba')=\text{wt}(\ba)+1=d(\overline{\cC_{f}})+1$.
Hence $d(\overline{\cC_f}')=d(\overline{\cC_{f}})+1$.\\

{Case 2:} Let $f(x)$ be balanced. Then $\text{W}_f(0)=0$ and $ 0\notin S_f$.\\

{Subcase 2.1:} Let $\varepsilon (\eta_0(-1))^{\frac{m+s}{2}}=1$.
 By Theorem \ref{weight},  we deduce $\text{wt}(\ba)=d(\overline{\cC_f})=p^m-p^{m-1}-(p-1)p^{\frac{m+s}{2}-1}$ if and only if $a \in \gf_p^*$ , $f^*(-\frac{b}{a})+a^{-1}(c+a)=0$ and $a^{-1}b\in S_f$.
Since $b\neq 0$, we have $\text{wt}(\ba')\geq \text{wt}(\ba)+2$. Hence  $d(\overline{\cC_f}')\geq d(\overline{\cC_{f}})+2$.\\

{Subcase 2.2:} Let $\varepsilon (\eta_0(-1))^{\frac{m+s}{2}}=-1$. By Theorem \ref{weight},  we deduce $\text{wt}(\ba)=d(\overline{\cC_f})=p^m-p^{m-1}-p^{\frac{m+s}{2}-1}$ if and only if $a \in \gf_p^*$ , $f^*(-\frac{b}{a})+a^{-1}(c+a)\neq0$ and $a^{-1}b\in S_f$.
Since $b\neq 0$, we also have $\text{wt}(\ba')\geq \text{wt}(\ba)+2$. Hence  $d(\overline{\cC_f}')\geq d(\overline{\cC_{f}})+2$.\\
\indent It is obvious that  that $\overline{\cC_f}'^{\perp}$ has length $q+m+2$ and dimension $q$.~We will prove that the minimum distance $d^\perp$ of $\overline{\cC_f}'^{\perp}$ is $3$ in the following.
Note that $G_1'$ is a check matrix of $\overline{\cC_f}'^{\perp}$ and any column vector in $G_1'$ is nonzero. Hence, $d^{\perp} \geq 2$.
Note that any two columns of the first $m+2$ columns of matrix $G_1'$ are $\gf_p$-linearly independent and any two columns of the last $q$ columns of matrix $G_1'$ are also $\gf_p$-linearly independent. Select any column $\bg_1$ in the first $m+2$ columns of matrix $G_1'$ and any column $\bg_2$ in the last $q$ columns of matrix $G_1'$. If $\bg_1 \neq (1,0,\cdots,0)$, it is easy to deduce that $\bg_1$ and $\bg_2$ are $\gf_p$-linearly independent. If $\bg_1 = (1,0,\cdots,0)$, then $\bg_1$ and $\bg_2$ are $\gf_p$-linearly dependent if and only if
\begin{eqnarray}\label{system}
\left\{\begin{array}{c}
  f(d_i)+1=0, \\
  \tr_{q/p}(\alpha^{0}d_{i})=0, \\
  \tr_{q/p}(\alpha^{1}d_{i})=0, \\
  \vdots \\
  \tr_{q/p}(\alpha^{m-1}d_{i})=0,
\end{array}\right.
\end{eqnarray}
where $d_i \in \gf_q$. This system of equations implies $d_i\in \gf_q^*$.
For any $x=\sum_{l=0}^{m-1}k_l\alpha^l \in \gf_{q}, k_l \in \gf_p$, the System (\ref{system}) implies
$$
\tr_{q/p}(xd_i)=\sum_{l=0}^{m-1}k_l \tr_{q/p}(\alpha^ld_i)=0.
$$
 This contradicts with the fact that $|\ker(\tr_{q/p})|=p^{m-1}$.
Hence, $\bg_1$ and $\bg_2$ are $\gf_p$-linearly independent and $d^{\perp} \geq 3$. Besides, the first, second and last columns of $G_1'$ are $\gf_p$-linearly dependent. Then $d^{\perp}=3$.
\end{proof}

Next we will consider the case that $m+s$ is odd.
\begin{theorem} \label{theorem2}
Let $q=p^m$, where $p$ is an odd prime, $m$ is a positive integer, $s$ is an integer and $m+s$ is odd with $m+s\geq 3$. Let $f(x)$ be weakly regular $s$-plateaued function with $\varepsilon$ the sign of the Walsh transform of $f(x)$
\begin{enumerate}[1)]
\item When $f(x)$ is unbalanced, then we have the following results.
 \begin{itemize}
 \item[1.1)]If~ $\varepsilon (\eta_0(-1))^{\frac{m+s+1}{2}}=1$, then the linear code $\overline{\cC_f}'$ has parameters
$$\left[p^m+m+2, m+2, p^m-p^{m-1}-p^{\frac{m+s-1}{2}}+1\right].$$
\item[1.2)] If~ $\varepsilon (\eta_0(-1))^{\frac{m+s+1}{2}}=-1$, then the linear code $\overline{\cC_f}'$ has parameters
$$\left[p^m+m+2, m+2, p^m-p^{m-1}-p^{\frac{m+s-1}{2}}+2\right].$$
\end{itemize}
\item  When $f(x)$ is balanced, then we have the following results.
 \begin{itemize}
\item[2.1)]If~ $\varepsilon (\eta_0(-1))^{\frac{m+s+1}{2}}=1$, then the linear code $\overline{\cC_f}'$ has parameters
$$\left[p^m+m+2, m+2, \geq p^m-p^{m-1}-p^{\frac{m+s-1}{2}}+2\right].$$
\item[2.2)] If~ $\varepsilon (\eta_0(-1))^{\frac{m+s+1}{2}}=-1$, then the linear code $\overline{\cC_f}'$ has parameters
$$\left[p^m+m+2, m+2, \geq p^m-p^{m-1}-p^{\frac{m+s-1}{2}}+2\right].$$
\end{itemize}
\end{enumerate}
\indent Moreover, the parameters of $\overline{\cC_f}'^{\perp}$ is given by $[p^m+m+2, p^m, 3]$ and $\overline{\cC_f}'^{\perp}$ is at least almost optimal with respect to the sphere-packing bound.
\end{theorem}
\begin{proof}
The proof is omitted as it is similar to that of  Theorem \ref{theorem1}.
\end{proof}

The following gives the optimally extendability of $\overline{\cC_f}$.
\begin{theorem} \label{theorem3}
Let $\overline{\cC_f}$ be the linear code in Theorem \ref{weight} and  $\overline{\cC_f}'$ be the linear code with generator matrix  $[I_{m+2, m+2}: G_1]$. Then $\overline{\cC_f}$  is an optimally extendable code with $m+s\geq 3$.
\end{theorem}

\begin{proof}
By Theorems \ref{weight},\ref{theorem1} and \ref{theorem2}, we have $d(\overline{\cC_{f}}^{\perp})-d(\overline{\cC_{f}}'^{\perp}) = 0$. Then the desired conclusion follows.
\end{proof}

\begin{remark}
Let $q=3^m$, where $m$ is a positive integer, $s$ is an integer such that $m+s\geq 3$.
It is easy to verify that the codes $\overline{\cC_f}'$ in Theorems \ref{theorem1} and \ref{theorem2} are leading-systematic ternary LCD codes.
We find that some ternary LCD codes are optimal according to the Code Tables at  http://www.codetables.de/.
They are listed in Table \ref{tab4}.
\end{remark}

\begin{table}[!h]
\begin{center}
\caption{Optimal  ternary LCD codes.}\label{tab4}
\begin{tabular}{llll}
\toprule
Conditions & Code & Parameters & Optimality \\
\midrule
$p=3, m=2$ & $\overline{\cC_f}'^{\perp}$ & $[13,9,3]$ & Optimal\\
$p=3, m=3$ & $\overline{\cC_f}'^{\perp}$ & $[32,27,3]$ & Optimal\\
$p=3, m=4$ & $\overline{\cC_f}'^{\perp}$ & $[87,81,3]$ & Optimal\\
\bottomrule
\end{tabular}
\end{center}
\end{table}

\section{Binary self-orthogonal codes and LCD codes from $s$-plateaued Boolean functions}
In this section, we study the code $\overline{\cC_{f}}$ defined in Equation (\ref{eq-Cfbar}) if $f(x)$ is an $s$-plateaued Boolean function.
We also construct binary LCD codes from $\overline{\cC_{f}}$.

\begin{theorem}\label{theorem4}
Let $q=2^m$, where  $0 \leq s \leq m-2$ is a nonnegative integer and $m$ is a positive integer with $m\geq 2$ such that $m+s$ is even. Let $f(x)$ be an $s$-plateaued Boolean function. The  $\overline{\cC_{f}}$ is a binary $[ q, m+2, 2^{m-1}-2^{\frac{m+s-2}{2}}]$ linear  code  and its weight distribution is listed in table \ref{tab3}. Particulary, $\overline{\cC_{f}}$ is self-orthogonal if $m+s\geq 6$. Besides, $\overline{\cC_{f}}^{\perp}$ has parameters $[q,q-m-2,4]$ and is optimal according to the sphere-packing bound.
\end{theorem}

\begin{table}\label{ta-3}
\begin{center}
\caption{The weight distribution of $\overline{\cC_{f}}$ in Theorem \ref{theorem4} when $p=2$ and $m+s$ is even.}\label{tab3}
\begin{tabular}{@{}ll@{}}
\toprule
Weight & Frequency  \\
\midrule
$0$ & $1$ \\
$2^{m-1}-2^{\frac{m+s-2}{2}}$ & $2^{m-s}$  \\
$2^{m-1}+2^{\frac{m+s-2}{2}}$ & $2^{m-s}$ \\
$2^{m-1}$ & $2^{m+2}-2^{m-s+1}-2$ \\
$2^m$ &   $1$\\
\bottomrule
\end{tabular}
\end{center}
\end{table}
\begin{proof}
Let $\bc_{{(a,b,c)}}=(af(x)+\tr_{q/p}(bx)+c)_{x \in \gf_q}, a \in \gf_2, b \in \gf_q, c \in \gf_2$, be a codeword in $\overline{\cC_f}$.
Then its Hamming weight is given by
\begin{eqnarray*}
  \text{wt}(\bc_{a,b,c})&=&2^m-\frac{1}{2}\sum_{y \in \gf_2}\sum_{x \in \gf_{2^{m}}}(-1)^{y(af(x)+\tr_{2^m/2}(bx)+c)}\\
   &=&  2^{m-1}-\frac{1}{2}\sum_{x \in \gf_{2^m}}(-1)^{af(x)+\tr_{2^m/2}(bx)+c}\\
   &=& \begin{cases}
         0 & \mbox{if $ a=0, b=0,c=0$,} \\
         2^m & \mbox{if $a=0, b=0, c=1$,}\\
         2^{m-1} & \mbox{if $a=0, b \neq 0, c\in \gf_2  \ \text{or} \ a \neq 0 , a^{-1}b\notin S_f, c\in \gf_2$,} \\
         2^{m-1}-\frac{1}{2}\text{W}_f(a^{-1}b) & \mbox{if $ a \neq 0, b\in \gf_{2^m}, c=0 , a^{-1}b \in S_f$,} \\
         2^{m-1}+\frac{1}{2}\text{W}_f(a^{-1}b) & \mbox{if $ a \neq 0, b\in \gf_{2^m}, c=1 , a^{-1}b \in S_f$.}
       \end{cases}
\end{eqnarray*}
According to Lemma \ref{plateaued-boolean}, we have $\text{W}_f(a^{-1}b) \in \{0, \pm 2^{(m+s)/2}\}$. Then the weight distribution directly follows where the frequency of each weight can be derived from Lemma \cite{Mesnager1}.
By Lemma \ref{lem-23}, we know $\overline{\cC_{f}}$ is self-orthogonal if $m+s\geq 6$.
The parameters of $\overline{\cC_{f}}^{\perp}$ can be derived by the Pless power moments in \cite[Page 260]{H}.
According to the sphere-packing bound, there exists no $[q,q-m-2,5]$ binary code. Then $\overline{\cC_{f}}^{\perp}$ is optimal with respect to the sphere-packing bound.
\end{proof}

In the following, we will construct a family of LCD codes from $\overline{\cC_{f}}$.\\
\indent Let $I_{m+2,m+2}$ be the identity matrix of size $m+2$ and $\gf_q^*=\langle\alpha \rangle$. Let $\{1,\alpha, \alpha^{2},\ldots, \alpha^{n-1}\}$ be a $\gf_2$-basis of $\gf_{q}$. Let $d_1,d_2,\cdots,d_{q-1}, d_q$ denote all elements of $\gf_q$, where $d_q=0$.
By definition, the generator matrix $G$ of $\overline{\cC_f}$ is given by
\begin{eqnarray*}
G:=\left[
\begin{array}{cccc}
1&1&\cdots&1 \\
 f(d_1) & f(d_2) & \cdots &f(d_q)\\
 \tr_{q/2}(\alpha^{0}d_{1})& \tr_{q/2}(\alpha^{0}d_{2})& \cdots &\tr_{q/2}(\alpha^{0}d_{q}) \\
\tr_{q/2}(\alpha^{1}d_{1})& \tr_{q/2}(\alpha^{1}d_{2})& \cdots &\tr_{q/2}(\alpha^{1}d_{q}) \\
\vdots &\vdots &\ddots &\vdots \\
\tr_{q/2}(\alpha^{m-1}d_{1})& \tr_{q/2}(\alpha^{m-1}d_{2})& \cdots &\tr_{q/2}(\alpha^{m-1}d_{q})
\end{array}\right].
\end{eqnarray*}
However, the dual of the LCD code ${\overline{\cC_f}'}$ generated by the matrix $G':=[I_{m+2,m+2}:G]$ has bad minimum distance $2$.
By an elementary row transformation on $G$, we obtain another generator matrix $G_2$ of $\overline{\cC_f}$  given by
\begin{eqnarray}\label{matrix2}
G_2:=\left[
\begin{array}{cccc}
1&1&\cdots&1 \\
 f(d_1)+1 & f(d_2)+1 & \cdots &f(d_q)+1\\
 \tr_{q/2}(\alpha^{0}d_{1})& \tr_{q/2}(\alpha^{0}d_{2})& \cdots &\tr_{q/2}(\alpha^{0}d_{q}) \\
\tr_{q/2}(\alpha^{1}d_{1})& \tr_{q/2}(\alpha^{1}d_{2})& \cdots &\tr_{q/2}(\alpha^{1}d_{q}) \\
\vdots &\vdots &\ddots &\vdots \\
\tr_{q/2}(\alpha^{m-1}d_{1})& \tr_{q/2}(\alpha^{m-1}d_{2})& \cdots &\tr_{q/2}(\alpha^{m-1}d_{q})
\end{array}\right].
\end{eqnarray}
Let ${\overline{\cC_f}'}$  be the linear code generate by $G_2':=[I_{m+2, m+2}: G_2].$
\begin{theorem}\label{theorem5}
Let $q=2^m$, where $m$ is a positive integer and $s$ is an integer such that $m+s$ is even with $m+s\geq 6$. The code $\overline{\cC_{f}'}$  generated by the matrix $G_2'=[I_{m+2, m+2}: G_2]$ is a leading-systematic binary LCD code whose parameters are given as follows. \\
\begin{enumerate}[(1)]
\item  If $f(x)$ is balanced, then the binary LCD code $\overline{\cC_f}'$ has parameters
$$\left[2^m+m+2, m+2, \geq 2^{m-1}-2^{\frac{m+s-2}{2}}+2\right].$$
 \item  Let $f(x)$ be unbalanced.
 \begin{itemize}
\item[(2.1)] If $\text{W}_f(0)=-2^{\frac{m+s}{2}}$,  then the binary LCD code $\overline{\cC_f}'$ has parameters
$$ \left[2^m+m+2, m+2, 2^{m-1}-2^{\frac{m+s-2}{2}}+1\right].$$
\item[(2.2)] If $\text{W}_f(0)=2^{\frac{m+s}{2}}$, then the binary LCD code $\overline{\cC_f}'$ has parameters
 $$ \left[2^m+m+2, m+2, 2^{m-1}-2^{\frac{m+s-2}{2}}+2\right].$$
 \end{itemize}
 \end{enumerate}
 Besides, $\overline{\cC_f}'^{\perp}$ has parameters $ \left[2^m+m+2, 2^m, 3\right]$ and is at least almost optimal according to the sphere-packing
bound.
\end{theorem}
\begin{proof}
By Lemma \ref{lem-LCD}, the linear code ${\overline{\cC_f}'}$ is a leading-systematic LCD code.~Obviously, the length of $\overline{\cC_f}'$ is $q+m+2$ and the dimension of $\overline{\cC_f}'$ is $m+2$. Similarly to the proof of Theorem \ref{theorem1}, we can deduce that the minimum distance $d$ of $\overline{\cC_f}'$ satisfies that $d \geq d(\overline{\cC_{f}})+1$, where $d(\overline{\cC_{f}})$ denotes  the minimum distance of $\overline{\cC_{f}}$. For a codeword $\ba'$ in $\overline{\cC_f}'$, we have
\begin{eqnarray*}
\ba'&=&(c,a,a_0,a_1 \cdots, a_{m-1})G_2'\\&=&(c,a,a_0, \cdots, a_{m-1}, af(d_1)+\sum_{i=0}^{m-1}a_i\tr_{q/2}(\alpha^{i}d_{1})+c+a, \cdots, af(d_q)+\sum_{i=0}^{m-1}a_i\tr_{q/2}(\alpha^{i}d_{q})+c+a),
\end{eqnarray*}
where $c, a, a_i \in \gf_2$ for all $0 \leq i \leq m-1$. Then $\ba'=(c,a,a_0,a_1 \cdots, a_{m-1}, \ba)$, where $$\ba=(c,a,a_0,a_1 \cdots, a_{m-1})G_2=(af(x)+\tr_{q/2}(bx) + c+a)_{x \in \gf_q}$$ is a codeword in $\overline{\cC_f}$, where $b:=\sum_{i=0}^{m-1}a_i\alpha^i$. \\
\indent By Theorem \ref{theorem4}, we deduce that $\wt(\ba)=d({\overline{\cC_f}})=2^{m-1}-2^{\frac{m+s-2}{2}}$ if and only if $ a \neq 0,~c+a=0 , \text{W}_f(a^{-1}b)=2^{\frac{m+s}{2}}$ and~$a^{-1}b \in S_f $, or $ a \neq 0, ~ c+a=1 , \text{W}_f(a^{-1}b)=-2^{\frac{m+s}{2}}$ and~$a^{-1}b \in S_f$.
\begin{enumerate}
\item If $f(x)$ is balanced, then $\text{W}_f(0)=0$. Hence $\text{W}_f(a^{-1}b)\neq0$ implies $b\neq0$. Then we deduce $d\geq 2^{m-1}-2^{\frac{m+s-2}{2}}+2$.
\item Let $f(x)$ be unbalanced, then $\text{W}_f(0)\neq 0$.
\begin{enumerate}
 \item
 Let $\text{W}_f(0)=-2^{\frac{m+s}{2}}$. Then $\wt(\ba)=d({\overline{\cC_f}})$ if  $a=1,b=c=0$.
 We then have $\wt(\ba')=d({\overline{\cC_f}})+1$ and $d=2^{m-1}-2^{\frac{m+s-2}{2}}+1$.
\item
Let $\text{W}_f(0)=2^{\frac{m+s}{2}}$.
Then we claim that $d\geq d({\overline{\cC_f}})+2$. Otherwise, if there exists a codeword $\ba'=(c,a,a_0,a_1 \cdots, a_{m-1}, \ba)$ of weight $d({\overline{\cC_f}})+1$.
Then $\wt(\ba)=d({\overline{\cC_f}})$ and $\wt(c,a,a_0,a_1 \cdots, a_{m-1})=1$.
Note that $\wt(\ba)=d({\overline{\cC_f}})$ implies $a\neq 0$ by Theorem \ref{theorem4}. Then $\wt(c,a,a_0,a_1 \cdots, a_{m-1})=1$ implies $c=b=0$.
Since $\text{W}_f(0)=2^{\frac{m+s}{2}}$, $\wt(\ba)=d({\overline{\cC_f}})$ implies $a=c=1$ which contradicts with $c=0$.
Then $d\geq d({\overline{\cC_f}})+2$. Let $a=c=1,b=0$. Then there exists a codeword $\ba'$ such that $\wt(\ba')=d({\overline{\cC_f}})+2$.
Thus $d=2^{m-1}-2^{\frac{m+s-2}{2}}+2$.
\end{enumerate}
\end{enumerate}
The parameters of $\overline{\cC_f}'^{\perp}$ can be derived by a similar method used in the proof of Theorem \ref{theorem1}.
\end{proof}

\begin{theorem} \label{theorem6}
 Let $m+s$ be even and $m+s\geq 6$. Let $\overline{\cC_f}$ be the linear code in Theorem \ref{theorem4} with generator matrix $G_2$ and  $\overline{\cC_f}'$ be the linear code with generator matrix  $[I_{m+2, m+2}: G_2]$. Then $\overline{\cC_f}$  is an almost optimally extendable code.
\end{theorem}
\begin{proof}
By Theorems \ref{theorem4} and \ref{theorem5}, we have $d(\overline{\cC_{f}}^{\perp})-d(\overline{\cC_{f}}'^{\perp}) = 1$. The desired conclusion follows from the definition of almost optimally extendable code.
\end{proof}

\section{Self-dual codes}
In this section, we derive some families of self-dual codes from some self-orthogonal codes.

The following is a simple result for self-orthogonal codes.
\begin{lemma}\label{lem-subcode}
  Let $\cC$ be a self-orthogonal code. Then for any $\cC' \subseteq \cC$, $\cC'$ is also a self-orthogonal code.
\end{lemma}

\begin{proof}
The desired conclusion directly follows from the definition of self-orthogonal code.
\end{proof}

\begin{lemma}\label{self26}
Let $G$ be a generator matrix of a linear code $\cC'$ and $\begin{bmatrix}
\textbf{1}\\
G
\end{bmatrix}$ be a generator matrix of a linear code $\cC$. If $\cC$ is self-orthogonal, then $\cC'$ is also self-orthogonal.
\end{lemma}
\begin{proof}
Note that $\cC' \subseteq \cC$.
Then the desired conclusion follows from Lemma \ref{lem-subcode}.
\end{proof}

\begin{theorem}\label{self27}
Let $q=p^m$, where $p$ is an odd prime, $m\geq 2$ is a positive integer and $0\leq s \leq m$ is an integer such that $m+s\geq 3$.
If $p\equiv 1\pmod{4}$ and $m\geq 2$ is any positive integer  or $p\equiv 3 \pmod{4}$ and $m$ is even,
then there exists a $\left[q-1,\frac{q-1}{2}\right]$ self-dual code over $\gf_p$ which is a subcode of ${\cC_f^*}^\perp$.
\end{theorem}

\begin{proof}

Define
$$\cC_f=\left\{(af(x)+\tr_{q/p}(bx))_{x \in \gf_q}: a \in \gf_p, b \in \gf_q\right\}.$$
Then $\cC_f$ is a subcode of $\overline{\cC_f}$ defined in Equation (\ref{eq-Cfbar}).
Since $\overline{\cC_f}$ is self-orthogonal by Theorem \ref{weight} if  $m+s\geq 3$, it follows from Lemma \ref{self26}
that $\cC_f$ is also self-orthogonal. Note that $\cC_f^*$ defined in Equation in (\ref{eq-Cf*})
can be derived by deleting the zero coordinate of each codeword of $\cC_f$.
Hence $\cC_f^*$ is self-orthogonal.
Then the desired conclusion follows from Lemma \ref{lem-dual}.
\end{proof}

\begin{theorem}\label{self-dual-even}
Let $q=2^m$, where  $0 \leq s \leq m-2$ is a nonnegative integer and $m$ is a positive integer with $m\geq 2$ such that $m+s\geq 6$ is even.
Let $\overline{\cC_f}$ be defined in Equation (\ref{eq-Cfbar}).
 Then there exists a $[q,\frac{q}{2}]$ binary self-dual code which is a subcode of ${\overline{\cC_f}^{\perp}} $.
 \end{theorem}

 \begin{proof}
 By Theorem \ref{theorem4}, $\overline{\cC_f}$ is self-orthogonal if $m+s\geq 6$ is even.
 Then the desired conclusion follows from Lemma \ref{lem-dual1}.
 \end{proof}

In the following, we give an example to illustrate how to find self-dual codes by Theorems \ref{self27} and \ref{self-dual-even}.

\begin{example}
\indent Let $\gf_{q}^*=\langle\alpha\rangle$ and $q=p^m$. Let $p=3$, $m=2$ and $f(x)=\tr_{3^m/3}({\alpha}x^4+{\alpha}^8x^2)$. Then $f(x)$ is weakly regular ternary 1-plateaued function and $\overline{\cC_f}$ is a $3$-weight
ternary $[9,4,3]$ code with weight enumerator $1+6z^{3}+66z^{6}+8z^{9}$. We know that $\overline{\cC_f}$ is self-orthogonal by Lemma \ref{lem-23}. It is easy to see that $\cC_f^*$ is also self-orthogonal with parameters $[8,3,3]$ by Lemma \ref{self26}. We will find a self-dual code $\widetilde{\cC}$ from $ {{\cC_{f}}^*}^{\perp}$.
Select a parity check matrix of $ {{\cC_{f}}^*}$ as follows:
 \begin{eqnarray}\label{matrix4}
H:=\left[
\begin{array}{cccc}
2 ~~~0 ~~~0~~~ 0~~~0 ~~~1 ~~~2~~~0\\
0 ~~~1 ~~~0~~~ 0 ~~~1 ~~~0 ~~~0 ~~~2\\
0 ~~~1 ~~~1 ~~~0~~~ 0 ~~~0 ~~~0~~~1\\
1 ~~~0 ~~~0 ~~~1 ~~~0 ~~~0~~~ 2~~~0\\
0~~~0~~~0 ~~~0 ~~~1 ~~~1 ~~~2~~~ 0
\end{array}\right].
\end{eqnarray}
Then $H$ is a generator matrix of ${{\cC_{f}}^*}^{\perp}$. Taking the first four rows of $H$ yields the following matrix:
\begin{eqnarray}\label{matrix4}
\widetilde{H}:=\left[
\begin{array}{cccc}
2 ~~~0 ~~~0~~~ 0~~~0 ~~~1 ~~~2~~~0\\
0 ~~~1 ~~~0~~~ 0 ~~~1 ~~~0 ~~~0 ~~~2\\
0 ~~~1 ~~~1 ~~~0~~~ 0 ~~~0 ~~~0~~~1\\
1 ~~~0 ~~~0 ~~~1 ~~~0 ~~~0~~~ 2~~~0\\
\end{array}\right].
\end{eqnarray}
Then $\widetilde{H}$ generate a $[8,4,3]$ ternary linear code $\widetilde{\cC}$ which is a subcode of ${{\cC_{f}}^*}^{\perp}$.
$\widetilde{\cC}$ has weight enumerator $1+16z^3+64z^6$.
Hence $\widetilde{\cC}$ is self-orthogonal by Lemma \ref{lem-23}. Since $\dim(\widetilde{\cC})=4$, $\widetilde{\cC}$ is self-dual.
\end{example}

\section{Summary and concluding remarks}
In this paper, we mainly studied the linear code ${\overline{\cC_f}}$ if $f$ is a weakly regular plateaued function or a plateaued Boolean function.
The highlights of this family of codes are as follows:
\begin{enumerate}
\item[$\bullet$] ${\overline{\cC_f}}$ is self-orthogonal under certain conditions (see Theorems \ref{weight} and \ref{theorem4});
\item[$\bullet$] The dual of ${\overline{\cC_f}}$ is optimal (see Theorem \ref{theorem4}) or at least almost optimal (see Theorem \ref{weight}) according to the sphere-packing bound;
\item[$\bullet$] ${\overline{\cC_f}}$ is an optimally extendable code (see Theorem \ref{theorem3}) or an almost optimally extendable code (see Theorem  \ref{theorem6}) by selecting a suitable generator matrix.
\end{enumerate}
Besides, we also derived some families of binary or ternary LCD codes from ${\overline{\cC_f}}$. In Table \ref{tab5}, we list known infinite families of binary or ternary LCD codes. From this table, our LCD codes have different parameters from known ones. In Table \ref{tab4}, these LCD codes are optimal in some cases.
Furthermore, we also derived self-dual codes from some self-orthogonal codes in this paper.
It is open to study the minimum distance of these self-dual codes.

 \begin{table}[h]
\scriptsize
\begin{center}
\caption{The known infinite families of binary or ternary LCD codes.}\label{tab5}
\setlength{\tabcolsep}{1mm}{
\begin{tabular}{llll} \hline
$n$ & $[n, k, d]$ & Condition & Reference \\ \hline
$n=3^m-1$ & $[n, 2m, \geq 2\delta_2]$ & $\delta_2=\frac{3^{m-1}-1}{4}+3^{m-2}$, $m$ odd & Th.4.5 in \cite{HY}\\
$n=3^m-1$ & $[n, 3, \frac{3^m-1}{2}]$ & $m$ even & Th.4.6 in \cite{HY}\\
$n=3^m-1$ & $[n, 2m, \geq 2\delta_3]$ & $m\equiv2$(mod 4), $\delta_3=\frac{3^{m-6}-1}{5}$ & Th.4.7 in \cite{HY}\\
 & &$+3^{m-6}+2 \cdot 3^{m-5}$\\
  & &$+2\cdot3^{m-3}+3^{m-2}$\\
$n=2^m-1$ & $[n, 2m, \geq 2\delta_1]$ & $m > 3$ odd, $\delta_1=\frac{2^{m-1}-1}{3}$ & Th.4.1 in \cite{HXY}\\
$n=2^m-1$ & $[n, 4m, \geq 2\delta_2]$ & $m > 3$ odd, $\delta_2=\frac{2^{m-1}-7}{3}$ & Th.4.1 in \cite{HXY}\\
$n=2^m-1$ & $[n, 6m, \geq 2\delta_3]$ & $m > 3$ odd, $\delta_3=\frac{2^{m-1}-25}{3}$ & Th.4.1 in \cite{HXY}\\
$n=2^m-1$ & $[n, 2m+2, \geq 2\delta_2]$ & $m \equiv 2 (\text{mod } 4)$, $\delta_2=\frac{2^{m-2}-1}{5}$ & Th.4.4 in \cite{HXY}\\
  & &$+2^{m-3}$\\
$n=2^m-1$ & $[n, 4m+2, \geq 2\delta_3]$ & $m \equiv 2 (\text{mod } 4)$, $\delta_3=\frac{2^{m-2}-31}{5}$ & Th.4.4 in \cite{HXY}\\
  & &$+2^{m-3}$\\
$n=2^m-1$ & $[n, 6,  2\delta_2]$ & $m \equiv 0 (\text{mod } 4)$, $\delta_2=\frac{2^{m}-1}{5}$ & Th.4.5 in \cite{HXY}\\
$n=q^\ell+1$ & $[n,k, \geq 2(\delta-1)]$ & $3 \leq \delta \leq q^{\lfloor(\ell-1)/2\rfloor}+3$ & Th.18 in \cite{LD}\\
& & $k=q^\ell-2\ell(\delta-2-\lfloor \frac{\delta-2}{q} \rfloor)$\\
$n=2^m-1$ & $[n,k, \geq 2(2^{m-\ell}-1)]$ & $m-2 \geq \ell \geq m-\lfloor(m-2)/2\rfloor$ & Th.26 in \cite{LD}\\
& & $k=2^m-2\sum_{j=0}^{m-1-\ell}\tbinom{m}{j}$\\
$n=\frac{q^m-1}{q-1}$ & $[n, n-1-2m\lceil\frac{(\delta-1)(q-1)}{q}\rceil, \geq 2\delta]$ & $2 \leq \delta \leq q^{\lfloor(m-1)/2\rfloor}$ & Th.32 in \cite{LD}\\
 $n=\frac{2^m+1}{3}$ & $[n,k, \geq \delta]$ & $m \geq 7$ odd and $3 \leq \delta \leq \frac{n}{2}$ & Th.3.1 in \cite{LKZ}\\
$n=\frac{3m}{2}$ & $[n, m, 2]$   &$m$ even & Th.10 in \cite{WY1} \\
$n=\frac{3m}{2}$ & $[n, \frac{m}{2}, 3]$ &$m$ even & Th.11 in \cite{WY1} \\
$n=\tbinom{m}{t}$ & $[n, n-m+1, 3]$ & $t$ even, $\tbinom{m-1}{t-1}\equiv 1 (\text{mod }2)$ & Th.25 in \cite{Z}\\
& & and $\tbinom{m-2}{t-2} \equiv 0 (\text{mod }2)$; or\\
& & $t$ even $, \tbinom{m-1}{t-1}\equiv\tbinom{m-2}{t-2}+1$\\
& &  $ \equiv 0 (\text{mod }2)$ and $tm \equiv 0 (\text{mod }2)$\\
$n=\tbinom{m}{t}$ & $[n, n-m, 4]$ & $t$ odd, $\tbinom{m-1}{t-1}\equiv 1 (\text{mod }2)$ & Th.25 in \cite{Z}\\
& &  and $\tbinom{m-2}{t-2} \equiv 0 (\text{mod }2)$; or\\
& & $t$ odd, $\tbinom{m-1}{t-1}\equiv \tbinom{m-2}{t-2} +1 $\\
& &  $\equiv 0 (\text{mod }2)$ and $tm \equiv 0 (\text{mod }2)$\\
$n=\tbinom{m}{t}+1$ & $[n, n-m+1, 3]$ & $m,t$ even, $\tbinom{m-1}{t-1}\equiv 0 (\text{mod }2)$ & Th.31 in \cite{Z}\\
& & and $\tbinom{m-2}{t-2} \equiv 1 (\text{mod }2)$; or\\
& & $m,t$ even $ \tbinom{m-1}{t-1}\equiv\tbinom{m-2}{t-2}$\\
& &  $ +1 \equiv 1 (\text{mod }2)$\\
$n=\tbinom{m}{t}+1$ & $[n, n-m, 4]$ & $t$ odd, $m \neq 2t$,  $\tbinom{m-1}{t-1}\equiv 0 (\text{mod }2)$ & Th.31 in \cite{Z} \\
& &  and $\tbinom{m-2}{t-2} \equiv 1 (\text{mod }2)$; or\\
& & $t$ odd, $m \neq 2t$, $\tbinom{m-1}{t-1}\equiv \tbinom{m-2}{t-2}$ \\
& &  $+1 \equiv 1 (\text{mod }2)$ and $m$ even\\
$n=n_f+m+1$ & $[n, m+1, \frac{n_f}{2}-2^{\frac{m-4}{2}}+2]$ &$S=0, m \geq 6$ even & Th.1 in \cite{WXY1} \\
& &  $n_f=2^{m-1} \pm 2^{\frac{m-2}{2}}$\\
$n=n_f+m+1$ & $[n, m+1, \frac{n_f}{2}-2^{\frac{m-4}{2}}+1]$ &$S>0, m \geq 6$ even & Th.1 in \cite{WXY1} \\
& &  $n_f=2^{m-1} \pm 2^{\frac{m-2}{2}}$\\
$n=n_f+m+1$ & $[n, n_f, 3]$ &$m \geq 6$ even & Th.1 in \cite{WXY1} \\
& &  $n_f=2^{m-1} \pm 2^{\frac{m-2}{2}}$\\
$n=q_0^{2r}+3r+1$ & $[n, 3r+1, (q_0-1)q_0^{2r-1}-q_0^{r-1}+2]$ &$r>1, (q_0,r)\neq(2,2)$ & Th.3 in \cite{WXY1} \\
$n=q_0^{2r}+3r+1$ & $[n, q_0^{2r}, 3]$ &$r>1, (q_0,r)\neq(2,2)$ & Th.3 in \cite{WXY1} \\
$n=p^m+m+2$       &$[n, m+2, p^m-p^{m-1}-(p-1)p^{\frac{m+s}{2}-1}+2]$ & $m+s\geq 4$ even, $f(x)$ unbalanced, $\varepsilon (\eta_0(-1))^{\frac{m+s}{2}}=1$ & Theorem \ref{theorem1} \\
$n=p^m+m+2$       &$[n, m+2, p^m-p^{m-1}-(p-1)p^{\frac{m+s}{2}-1}+2]$ &$m+s\geq 4$ even, $f(x)$ unbalanced, $\varepsilon (\eta_0(-1))^{\frac{m+s}{2}}=-1$ & Theorem \ref{theorem1} \\
$n=p^m+m+2$       &$[n, m+2, \geq p^m-p^{m-1}-(p-1)p^{\frac{m+s}{2}-1}+2]$ &$m+s\geq 4$ even, $f(x)$ balanced, $\varepsilon (\eta_0(-1))^{\frac{m+s}{2}}=1$
& Theorem \ref{theorem1} \\
$n=p^m+m+2$       &$[n, m+2, \geq p^m-p^{m-1}-p^{\frac{m+s}{2}-1}+2]$
&$m+s\geq 4$ even, $f(x)$ balanced, $\varepsilon (\eta_0(-1))^{\frac{m+s}{2}}=-1$
& Theorem \ref{theorem1} \\
$n=p^m+m+2$       &$[n, m+2, p^m-p^{m-1}-p^{\frac{m+s-1}{2}}+1]$ &$m+s\geq 3$ odd, $f(x)$ unbalanced, $\varepsilon (\eta_0(-1))^{\frac{m+s}{2}}=1$
& Theorem \ref{theorem2} \\
$n=p^m+m+2$       &$[n, m+2, p^m-p^{m-1}-p^{\frac{m+s-1}{2}}+2]$ &$m+s\geq 3$ odd, $f(x)$ unbalanced, $\varepsilon (\eta_0(-1))^{\frac{m+s}{2}}=-1$ & Theorem \ref{theorem2} \\
$n=p^m+m+2$       &$[n, m+2, \geq p^m-p^{m-1}-p^{\frac{m+s-1}{2}}+2]$&$m+s\geq 3$ odd, $f(x)$ balanced & Theorem \ref{theorem2} \\
$n=p^m+m+2$       &$[n, p^m, 3]$ &$m+s\geq 3$ & Theorem \ref{theorem1}, Theorem \ref{theorem2}\\
$n=2^m+m+2$       &$[n, m+2, \geq 2^{m-1}-2^{\frac{m+s-2}{2}}+2]$& $m+s\geq 6$ even, $f(x)$ balanced & Theorem \ref{theorem5} \\
$n=2^m+m+2$       &$[n, m+2, 2^{m-1}-2^{\frac{m+s-2}{2}}+1]$& $m+s\geq 6$ even, $f(x)$ unbalanced,$\text{W}_f(0)=-2^{\frac{m+s}{2}}$ & Theorem \ref{theorem5} \\
$n=2^m+m+2$       &$[n, m+2, 2^{m-1}-2^{\frac{m+s-2}{2}}+2]$& $m+s\geq 6$ even, $f(x)$ unbalanced,$\text{W}_f(0)=2^{\frac{m+s}{2}}$ & Theorem \ref{theorem5} \\
$n=2^m+m+2$       &$[n, 2^m, 3]$ & $m+s\geq 6$ even& Theorem \ref{theorem5}\\
\hline\end{tabular}}
\end{center}
\end{table}

\end{document}